\documentclass{llncs}

\usepackage{amsmath}
\usepackage{amssymb}
\usepackage{algorithm}
\usepackage{graphicx}
\usepackage{times}
\usepackage{float}

\floatstyle{ruled}
\newfloat{algo}{thp}{lop}
\floatname{algo}{Algorithm}

\newfloat{eqncap}{thp}{lop}
\floatname{eqncap}{Equation}

\newfloat{examplecap}{thp}{lop}
\floatname{examplecap}{Example}

\newcommand{\ceil}[1]{\lceil{#1}\rceil}
\newcommand{\floor}[1]{\lfloor{#1}\rfloor}

\newcommand{\Zp}{\mathbb{Z}_p}

\begin{document}

\title{Fast Pseudo-Random Fingerprints}

\author{Yoram Bachrach\inst{1}, Ely Porat\inst{2}}
\institute{
Microsoft Research, Cambridge, UK \email{(yobach@micorosft.com)}
\and
Bar-Ilan University, Ramat Gan, Israel \email{(porately@cs.biu.ac.il)} \\
}

%\author{}
%\institute{}

\maketitle

\begin{abstract}
We propose a method to exponentially speed up computation of various fingerprints, such as the ones used to compute similarity and rarity in massive data sets. Rather then maintaining the full stream of $b$ items of a universe $[u]$, such methods only maintain a concise fingerprint of the stream, and perform computations using the fingerprints. The computations are done approximately, and the required fingerprint size $k$ depends on the desired accuracy $\epsilon$ and confidence $\delta$. Our technique maintains a single bit per hash function, rather than a single integer, thus requiring a fingerprint of length $k = O(\frac{\ln \frac{1}{\delta}}{\epsilon^2})$ bits, rather than $O(\log u \cdot \frac{\ln \frac{1}{\delta}}{\epsilon^2})$ bits required by previous approaches. The main advantage of the fingerprints we propose is that rather than computing the fingerprint of a stream of $b$ items in time of $O(b \cdot k)$, we can compute it in time $O(b \log k)$. Thus this allows an exponential speedup for the fingerprint construction, or alternatively allows achieving a much higher accuracy while preserving computation time. Our methods rely on a specific family of pseudo-random hashes %that was shown to be approximately-MWIF~\cite{indyk2001small}, and 
for which we can quickly locate hashes resulting in small values. 
%Our analysis uses the fact that members of the family behave in a pairwise independent manner. , altohugh they are not all completely independent.  
\end{abstract}     

\section{Introduction}
\label{sec:intro}  

%%%
%maybe to move some of it to the abstract as well.
%You must emphasys "Exponantioal improvment", "General technieque work as well for shopisticated hash function like min wise independent" (In almost all other applications I don't need min wise independent
Hashing is a key tool in processing massive data sets. Many uses of hashing in various applications require computing many hash functions in parallel. In this paper we present a technique that ``ties together'' many hashes in a novel way, which enables us to speed up such algorithms by an \emph{exponential factor}. Our method also works for some complicated hash function such as min-wise independent families of hashes. In this paper we focus on producing an optimal similarity fingerprint using this method, but our technique is \emph{general}, as it is easy to use our approach to speed up other hash intensive computations.
%, such as the construction of many other fingerprints. 
One easy example where our technique applies is approximating the number of distinct elements from~\cite{alon1999space}. A another example, which requires a slightly stronger analysis, is computing of $L_p$ sketches~\cite{indyk2006stable} for $0\le p\le 2$ .  

Min-wise independent families of hash functions, which we call \emph{MWIFs} for short, were introduced in~\cite{mulmuley1996randomized,broder2000min}. Computations using MWIFs have been used in many algorithms for processing massive data streams. The properties of MWIFs 
%and their resemblance to random permutations 
allow maintaining concise descriptions of massive streams. These descriptions, called ``fingerprints'' or ``sketches'', allow computing properties of these streams and relations between them. Examples of such ``fingerprint'' computations include data summerization and subpopulation-size queries~\cite{cohen2007summarizing,cohen2007sketching}, greedy list intersection~\cite{krauthgamer2008greedy}, approximating rarity and similarity for data streams~\cite{datar2002estimating}, collaborative filtering fingerprints~\cite{bachrach2009sketching,bachrach2009spire,bachrach-fingerprinting} and estimating frequency moments~\cite{alon1999space}. %\cite !!!
Another motivation for studying MWIFs is reducing the amount of randomness used by algorithms~\cite{broder2003derandomization,mulmuley1996randomized,broder2000min}. 

Recent research reduced the amount of information stored, while accurately computing properties data streams. Such techniques improve the \emph{space complexity}, but much less attention has been given to \emph{computation complexity}. For example, many streaming algorithms compute huge amounts of hashes, as they apply \emph{many} hashes to each element in a very long stream of elements. This leads to a high computation time, not always tractable for many applications. 

Our main contribution is a method allowing an \emph{exponential} speedup in \emph{computation time} for constructing fingerprints of massive data streams. Our technique is 
\emph{general}, and can speed up many processes that apply many random hashes. The heart of the method lies in using a specific family of pseudo-random hashes shown to be approximately-MWIF~\cite{indyk2001small}, and for which we can quickly locate the hashes resulting in a small value of an element under the hash. Similarly to~\cite{patrascu:k} we use the fact that members of the family are pairwise independent between themselves. 
%%, and uses several such families to achieve a desired accuracy and confidence. 
%Although our fast fingerprint computation method may introduce a small error, we show that by increasing the fingerprint size we can easily overcome this error, and achieve the desired accuracy and confidence.  
We also extend the technique and show one can maintain just a \emph{single} bit rather than the full element IDs, thus improving the fingerprint size. Independently of us~\cite{li2010b} also considered storing few bits per hash function, but focused only on minimizing storage rather than computation time.

% XXXXX THE ERROR THAT WE ADD IS VERY SMALL IF WE TAKE THE $c\log k$ MINIMAL VALUES THE CHANGE IN THE 7/8 WILL BE $(ck)^{-O(c\log\log k)}$ WHICH IS REALLY SMALL
% TODO: This should not go in the intro. Where should it go?
%XXXXXX ADD HERE THE SKETCH SIZE IMPROVMENT.

%The paper proceeds as follows. Section~\ref{sec:perlim} provides background. Section~\ref{l_sect_hash_family} presents the hashes we use and their properties. We discuss the example of Jackard similarity in Section~\ref{l_sect_pairwise_block_const_conf}, using an analysis suitable for the hashes we use.
%%and~\ref{l_sect_jackard_req_conf}. Section~\ref{l_sect_pairwise_block_const_conf} shows how to achieve a desired accuracy with a constant confidence using the suggested hash family, and Section~\ref{l_sect_jackard_req_conf} shows how to achieve a target confidence. 
%Section~\ref{l_sect_fast_compute} discusses our exponential speedup of the fingerprint computation.%, and shows that by slightly increasing the fingerprint size, we overcome the small errors our method may introduce. 
%We conclude in Section~\ref{sec:conclusions}.%, which discusses how the method can to other general fingerprints. 

%%%

\subsection{Preliminaries}
\label{sec:perlim}  

Let $H$ be a family of functions over the same source $X$ and target $Y$, so each $h \in H$ is a function $h : X \rightarrow Y$, where $Y$ is a completely ordered set.
We say that $H$ is min-wise independent if, when randomly choosing a function $h \in H$, for any subset $C \subseteq X$, any $x \in C$ has an equal probability of being the minimal after applying $h$.
%\footnote{We must choose $h$ from $H$ under a certain  distribution. We will assume that $h$ is chosen uniformly from $H$, although any distribution for choosing a member of $h$ that would make $H$ min-wise independent would do.}

\begin{definition}{$H$ is min-wise independent} (MWIF), if for all $C \subseteq X$, for any $x \in C$, 
$Pr_{h \in H}[h(x) = min_{a \in C}h(a)] = \frac{1}{|C|}$
\end{definition}

\begin{definition}{$H$ is a $\gamma$-approximately min-wise independent} ($\gamma$-MWIF), if for all $C \subseteq X$, for any $x \in C$, 
$ \left| Pr_{h \in H}[h(x) = min_{a \in C}h(a)] - \frac{1}{|C|} \right| \leq \frac{\gamma}{|C|}$
\end{definition}

\begin{definition}{$H$ is $k$-wise independent}, if for all $x_1,x_2,\ldots,x_k,y_1,y_2,\ldots,y_k \subseteq X$,
$Pr_{h \in H}[ (h(x_1) = y_1) \wedge \ldots \wedge (h(x_k) = y_k) ] = \frac{1}{|X|^k}$
\end{definition}

%%%%
\section{Pseudo-Random Family of Hashes}
\label{l_sect_hash_family}
% We use a sepcific family of $\gamma$-MWIF/PWIF family, that 
We describe the hashes we use.%, to allow us to speed up the computation. 
Given the universe of item IDs $[u]$, consider a big prime $p$, such that $p>u$. Consider taking random coefficients for a $d$-degree polynomial in $\Zp$. Let $a_0, a_1, \ldots, a_d \in [p]$ be chosen uniformly at random from $[p]$, and the following polynomial in $\Zp$: $f(x) = a_0 + a_1 x + a_2 x^2 + \ldots + a_d x^d$. We denote by $F_d$ the family of all $d$-degree polynomials in $\Zp$ with coefficients in $\Zp$, and later choose members of this family uniformly at random. Indyk~\cite{indyk2001small} shows that choosing a function $f$ from $F_d$ uniformly at random results in $F_d$ being a $\gamma$-MWIF for $d = O(\log\frac{1}{\gamma})$. 
% TODO: Ely - fill the citation, and maybe make it a proper theorem?

Randomly choosing $a_0,\ldots,a_d$ is equivalent to choosing a member of $F_d$ uniformly at random, so $f(x) = a_0 + a_1 x + a_2 x^2 + \ldots + a_d x^d$ is a hash chosen at random from the $\gamma$-MWIF $F_d$. Similarly, consider $b_0, b_1, \ldots, b_d \in [p]$ be chosen uniformly at random from $[p]$, and $g(x) = b_0 + b_1 x + b_2 x^2 + \ldots + b_d x^d$, which is also a hash chosen at random from the $\gamma$-MWIF $F_d$. Now consider the hashes $h_0(x) = f(x), h_1(x) = f(x) + g(x), h_2(x) = f(x) + 2 g(x), \ldots, h_i(x) = f(x) + i g(x), \ldots, h_{k-1}(x) = f(x) + (k-1) g(x)$. We call this random construction procedure for $f(x),g(x)$ the \emph{base random construction}, and the construction of $h_i$ the \emph{composition construction}. We prove properties of such hashes. We denote the probability of an event $E$ when the hash $h$ is constructed by choosing $f,g$ using the base random construction and composing $h(x) = f(x) + i \cdot g(x)$ (for some $i \in [p]$) as $Pr_{h}(E)$. 

\begin{lemma}[Uniform Minimal Values]
Let $f,g$ be constructed using the base random construction, using $d = O(\log\frac{1}{\gamma})$. 
For any $z \in [u]$, any $X \subseteq [u]$ and any value $i$ used to compose $h(x)=f(x)+i \cdot g(x)$: %, the following holds: 
$Pr_{h} [h(z) < min_{y\in X} (h(y)] = (1\pm \gamma) \frac{1}{|X|}$. 
\end{lemma}
\begin{proof}
Fix $i$, $z \in [u]$ and $X \subseteq [u]$, construct $f,g$ using the base random construction, and compose $h(x)=f(x)+i \cdot g(x)$. Note in $i\cdot g(x) = i \cdot (b_0 + b_1 x + \ldots b_d x^d)$,  the coefficient of $x^j$ is $q = (i \cdot b_j) \mod p$. Given a value $s \in [p]$ There is exactly one value in $r \in [p]$ such that $(q+r) \mod p = s$. Thus, for any $s \in [p]$, the probability that the coefficient of $x^j$ in $h(x)$ is $s$ is $\frac{1}{p}$. Therefor $Pr_h[h(x)\equiv p(x)]=\frac{1}{p^{d+1}}=\frac{1}{F_d}$.
%\begin{align*}
% &Pr_{h} [h(z) < min_{y\in X} h(y)]=\\
% &\sum_{p(x) \in F_d} Pr_{h} [      h(z) < min_{y\in X} h(y) | h(x) \equiv f(x) + i \cdot g(x)\equiv p(x)] \cdot Pr_{h}[h(x) \equiv p(x)] = \\ 
% &\sum_{p(x) \in F_d} \frac{Pr_{h} [h(z) < min_{y\in X} h(y) | h(x) \equiv p(x)] }{|F_d|}= \sum_{p(x) \in F_d} \frac{Pr[p(z) < min_{y\in X} (p(y))] }{|F_d|}=\\
% &Pr_{p(x)   \in F_d}                      [p(z) < min_{y\in X} (p(y))]=(1\pm \gamma) \frac{1}{|X|}
%\end{align*}
We have:
$Pr_{h} [h(z) < min_{y\in X} h(y)]= 
\sum_{p(x) \in F_d} Pr_{h} [      h(z) < min_{y\in X} h(y) | h(x) \equiv f(x) + i \cdot g(x)\equiv p(x)] \cdot Pr_{h}[h(x) \equiv p(x)] = 
\sum_{p(x) \in F_d} \frac{Pr_{h} [h(z) < min_{y\in X} h(y) | h(x) \equiv p(x)] }{|F_d|}= \sum_{p(x) \in F_d} \frac{Pr[p(z) < min_{y\in X} (p(y))] }{|F_d|}=
Pr_{p(x)   \in F_d} [p(z) < min_{y\in X} (p(y))]=(1\pm \gamma) \frac{1}{|X|}$.
If $p(x)$ is a polynom such that for any $z \in \Zp$ we have $p(z) < min_{y\in X} (p(y))$, then we have $Pr[p(z) < min_{y\in X} (p(y))]=1$, and otherwise $Pr[p(z) < min_{y\in X} (p(y))]=0$. Thus we get $\sum_{p(x) \in F_d} \frac{Pr[p(z) < min_{y\in X} (p(y))] }{|F_d|}= Pr_{p(x)   \in F_d} [p(z) < min_{y\in X} (p(y))]$.
The last transition uses the fact that $F_d$ is an $\gamma$-MWIF, which requires $d = O(\log\frac{1}{\gamma})$. 
\end{proof}

\begin{lemma}[Pairwise Interaction]
\label{l_lem_pairwise_interact}
Let $f,g$ be constructed using the base random construction, using $d = O(\log\frac{1}{\gamma})$. 
For all $x_1,x_2 \in [u]$ and all $X_1,X_2 \subseteq [u]$, and all $i \neq j$ used to compose $h_i(x)=f(x)+i \cdot g(x)$ and $h_j(x)=f(x)+j\cdot g(x)$: %, the following holds:
$$Pr_{f,g \in F_d} [ (h_i(x_1) < min_{y\in X_1} h_i(y)) \wedge (h_j(x_2) < min_{y\in X_2} h_i(y)) ] =(1\pm \gamma)^2\frac{1}{|X_1| \cdot |X_2|}$$
%$$Pr_{f,g  F} [f(x_1)+ig(x_1)<min_{y\in X_1} (f(y)+ig(y)) \wedge f(x_2)+jg(x_2)<min_{y\in X_2} (f(y)+jg(y))]=(1\pm \epsilon)^2\frac{1}{X_1 X_2}$$
\end{lemma}
\begin{proof}
Given $p_1(x) \in F_d = u_0 + u_1 x + \ldots + u_d x^d$ and $p_2(x)  \in F_d = v_0 + v_1 x + \ldots + v_d x^d$, there is \emph{exactly one} pair of polynoms $f(x),g(x) \in F_d$ such that both $f(x)+i \cdot g(x) = p_1(x)$ and $f(x) + j \cdot g(x) = p_2(x)$. Each coefficient location $l \in [d]$ results in two equations with two unknowns in $\Zp$, with a single solution $(a_l,b_l)$ (where $a_l$ is the coefficient of $x^l$ in $f(x)$, and $b_l$ is the coefficient of $x^l$ in $g(x)$. 

Fix $i \neq j$, $x_1,x_2 \in [u]$ and $X_1,X_2 \subseteq [u]$, construct $f,g$ using the base random construction, and compose $h_i(x)=f(x)+i \cdot g(x)$, $h_j(x)=f(x)+j \cdot g(x)$.
For brevity, denote $m^i_1 = \min_{y\in X_1} h_i(y)$. Similarly, denote $m^j_2 = \min_{y\in X_2} h_j(y)$.
%\begin{align*}
% & Pr_{f,g \in F_d} [ (h_i(x_1) < m^i_1) \wedge (h_j(x_2) < m^j_2)] \\
% &=\sum_{p_1,p_2 \in F_d}              Pr [     (h_i(x_1) < m^i_1)  \wedge  (h_j(x_2) < m^j_2) | (h_i(x)\equiv p_1(x) \wedge h_j(x)\equiv p_2(x))] \\
% &\ \ \ \ \ \ \cdot Pr [(h_i(x)\equiv p_1(x) \wedge h_j(x)\equiv p_2(x))] \\
% &=\sum_{p_1,p_2 \in F_d}              \frac{Pr[(h_i(x_1) < m^i_1)  \wedge  (h_j(x_2) < m^j_2) | (h_i(x)\equiv p_1(x) \wedge h_j(x)\equiv p_2(x))]}{|F_d|^2} \\
% &=\sum_{p_1,p_2 \in F_d}              \frac{Pr[(p_1(x_1) < m^i_1)  \wedge  (p_2(x_2) < m^j_2)]}{|F_d|^2} =\sum_{p_1,p_2 \in F_d}              \frac{Pr[ p_1(x_1) < m^i_1 ] \cdot Pr[p_2(x_2) < m^j_2]}{|F_d|^2} \\
% &=\sum_{p_1\in F_d} \sum_{p_2\in F_d} \frac{Pr[ p_1(x_1) < m^i_1 ]}{|F_d|} \cdot \frac{Pr[p_2(x_2) < m^j_2]}{|F_d|} =(1\pm \gamma)^2\frac{1}{|X_1| \cdot |X_2|} 
We have:
$Pr_{f,g \in F_d} [ (h_i(x_1) < m^i_1) \wedge (h_j(x_2) < m^j_2)] =
\sum_{p_1,p_2 \in F_d}              Pr [     (h_i(x_1) < m^i_1)  \wedge  (h_j(x_2) < m^j_2) | (h_i(x)\equiv p_1(x) \wedge h_j(x)\equiv p_2(x))] \cdot Pr [(h_i(x)\equiv p_1(x) \wedge h_j(x)\equiv p_2(x))] =
\sum_{p_1,p_2 \in F_d}              \frac{Pr[(h_i(x_1) < m^i_1)  \wedge  (h_j(x_2) < m^j_2) | (h_i(x)\equiv p_1(x) \wedge h_j(x)\equiv p_2(x))]}{|F_d|^2} 
$.
Thus, 
$Pr_{f,g \in F_d} [ (h_i(x_1) < m^i_1) \wedge (h_j(x_2) < m^j_2)] =
\sum_{p_1,p_2 \in F_d}              \frac{Pr[(p_1(x_1) < m^i_1)  \wedge  (p_2(x_2) < m^j_2)]}{|F_d|^2} =\sum_{p_1,p_2 \in F_d}              \frac{Pr[ p_1(x_1) < m^i_1 ] \cdot Pr[p_2(x_2) < m^j_2]}{|F_d|^2} =
\sum_{p_1\in F_d} \sum_{p_2\in F_d} \frac{Pr[ p_1(x_1) < m^i_1 ]}{|F_d|} \cdot \frac{Pr[p_2(x_2) < m^j_2]}{|F_d|} =
(1\pm \gamma)^2\frac{1}{|X_1| \cdot |X_2|} $
\end{proof}

\section{Fingerprinting Using Pseudo-Random Hashes}
\label{l_sect_pairwise_block_const_conf}

Several methods were suggested for building fingerprints for approximating relations between massive datasets, such as the Jackard similarity (see~\cite{broder2000min} for example). 
%Consider approximating the Jackard similarity between 
Given a universe $U$, where $|U|=u$, consider $C_1,C_2$, where each $C_i \subseteq U$ is described as a set $|C_i|$ integers in $[u]$ (we use $[u]$ to denote $\{1,2,\ldots,u\})$. The Jackard similarity is $J_{1,2}=\frac{|C_1 \cap C_2|}{|C_1 \cup C_2|}$. %Fingerprinting methods typically rely on using MWIF hashes, and use Hoeffding's inequality to provide approximation bounds. 
Many fingerprints rely on 
%A shortcoming of such methods is that they must 
applying many hashes to each elements in the long streams. We use a the hashes of Section~\ref{l_sect_hash_family} to exponentially speed up such computations. We use pseudo-random effects in this hash, so we must relax the MWIF requirement to a pairwise independence requirement (2-wise independence). 
%Our analysis relies on both the Hoeffding and Chebychev bounds. 

For completeness, we briefly consider previously suggested approaches for approximating Jackard similarity~\cite{broder2000min}. Let $h \in H$ be a randomly chosen function from a MWIF $H$. We can apply $h$ on all elements $C_1$ and examine the minimal integer we get, $m^h_1 = \arg \min_{x \in C_1} h(x)$. We can do the same to $C_2$ and examine $m^h_2 = \arg \min_{x \in C_2} h(x)$. Fingerprints for estimating the Jackard similarity are based on computing the probability that $m_1 = m_2$:
$Pr_{h \in H} [m^h_1 = m^h_2] = Pr_{h \in H} [\arg \min_{x \in C_1} h(x) = \arg \min_{x \in C_2} h(x)]$. 

\begin{theorem}[Jackard and MWIF Collision Probability]
$Pr_{h \in H} [m^h_i = m^h_j] = J_{i,j}$.
\label{l_thm_prob_collision_jackard}
The proof is given in~\cite{broder2000min}, and in the appendix for completeness. 
\end{theorem}
%\begin{proof}
%\end{proof}

Similarly, regarding a hash $h$ from a $\gamma$-MWIF, \cite{broder1998resemblance,broder2000min} shows that: 
\begin{theorem}
$|Pr_{h \in H} [m^h_i = m^h_j] - J_{i,j}| \leq \gamma$.
\label{l_thm_prob_collision_jackard_approx_mwif}
\end{theorem}

Rather than maintaining the full $C_i$'s, previous approaches~\cite{broder1998resemblance,broder2000min} suggest maintaining their fingerprints. Given $k$ hashes $h_1,\ldots,h_k$ randomly chosen from an $\gamma$-MWIF, we can maintain $m^{h_1}_i,\ldots,m^{h_k}_i$. Given $C_i, C_j$, for any $x \in [k]$, the probability that $m^{h_x}_i = m^{h_x}_j$ is $J_{i,j} \pm \gamma$. A hash $h_x$ where we have $m^{h_x}_i = m^{h_x}_j$ is called a hash collision. We can thus estimate $J$ by counting the proportion of collision hashes out of all the chosen hashes. In this approach, the fingerprint contains $k$ item identities in $[u]$, since for any $x$, $m^{h_x}_i$ is in $[u]$. Thus, such a fingerprint requires $k \log u$ bits. To achieve an accuracy $\epsilon$ and confidence $\delta$, such approaches require $k = O ( \frac{\ln \frac{1}{\delta}}{\epsilon^2} )$. 
% Maybe cite a full result / analysis, 
Our basis for the fingerprint is a ``block fingerprint'' which allows approximating $J_{i,j}$ with a given accuracy $\epsilon$ and a confidence of $\frac{7}{8}$. This block fingerprint maintains only a \emph{single bit} per hash, as opposed to previous approaches which maintain $\log u$ bits per hash. % TODO , elaborate. 
Later we show how to achieve a given accuracy $\epsilon$ with a given confidence $\delta$, by combining several block fingerprints, and creating a full fingerprint. 

To shorten the fingerprints using a single bit per hash, we use a hash mapping elements in $[u]$ to a single bit --- $\phi : [u] \rightarrow \{0,1\}$, taken from a pairwise independent family (PWIF for short) of such hashes. Rather than defining $m^h_i = \arg \min_{x \in C_1} h(x)$ we define $m^{\phi,h}_i = \phi(\arg \min_{x \in C_1} h(x))$. Maintaining $m^{\phi,h}_i$ rather than $m^\phi_i$ shortens the fingerprint by a factor of $\log u$. We examine the resulting accuracy and confidence.
%consequences of doing so in terms of accuracy and confidence, or the required number of hashes to achieve a given accuracy and confidence. 

\begin{theorem}
$Pr_{h \in H} [m^{\phi,h}_i = m^{\phi,h}_j] = \frac{J_{i,j}}{2} + \frac{1}{2} \pm \frac{\gamma}{2}$.
\label{l_thm_prob_collision_jackard_approx_mwif_single_bit}
\end{theorem}
\begin{proof}
$Pr_{h\in H,\phi\in H'}[m_i^{\phi,h} = m_j^{\phi,h}] =
      Pr[m_i^{\phi,h} = m_j^{\phi,h} | m_i^h = m_j^h] \cdot    Pr_{h\in H}[m_i^{h} =    m_j^{h}]+
      Pr[m_i^{\phi,h} = m_j^{\phi,h} | m_i^h \neq m_j^h] \cdot Pr_{h\in H}[m_i^{h} \neq m_j^{h}]=
      1 \cdot Pr_{h\in H}[m_i^{h} = m_j^{h}]+\frac{1}{2}\cdot(1-Pr_{h\in H}[m_i^{h} = m_j^{h}])=\frac{1+J_{i,j}\pm\gamma}{2}$
\end{proof}

% TODO: say something about splitting the epsilon error
The purpose of the fingerprint block is to provide an approximation of $J$ with accuracy $\epsilon$. We use $k$ hashes, and choose $k=\frac{8.02}{\epsilon^2}$.
Denote $\alpha = \frac{2^{10}-1}{2^{10}}$, and let $\gamma = (1-\alpha) \cdot \epsilon = \frac{1}{2^{10}} \epsilon$. We construct a $\gamma$-MWIF \footnote{The accuracy $\gamma$ is much stronger than the overall accuracy $\epsilon$ required of the full fingerprint, for reasons to be later examined}. To construct the family, consider choosing $a_0,\ldots,a_d$ and $b_0, b_1, \ldots, b_d$ uniformly at random from $[p]$, constructing the polynomials $f(x) = a_0 + a_1 x + a_2 x^2 + \ldots + a_d x^d$, $g(x) = b_0 + b_1 x + b_2 x^2 + \ldots + b_d x^d$, and using the $k$ hashes $h_i(x) = f(x) + i g(x)$, where $i \in \{0,1,\ldots,k-1\}$. We also use a hash $\phi : [u] \rightarrow \{0,1\}$ chosen from the PWIF of such hashes. We say there is a collision on $h_l$ if $m^{\phi,h_l}_i = m^{\phi,h_l}_j$, and denote the random variable $Z_l$ where $Z_l = 1$ if there is a collision on $h_l$ for users $i,j$ and $Z_l=0$ if there is no such collision. $Z_l = 1$ with probability $\frac{1}{2} + \frac{J}{2} \pm \frac{\gamma}{2}$ and $Z_l = 0$ with probability $\frac{1}{2} - \frac{J}{2} \pm \frac{\gamma}{2}$. Thus $E(Z_l) = \frac{1}{2} + \frac{J}{2} \pm \frac{\gamma}{2}$. Denote $X_l = 2 Z_l - 1$. $E(X_l) = 2 E(Z_l) - 1 = J \pm \gamma$. $X_l$ can take two values, $-1$ when $Z_l=0$, and $1$ when $Z_l=1$. Thus $X_l^2$ always takes the value of $1$, so $E(X_l^2)=1$. 
Consider $X=\sum_{l=1}^k X_l$, and take $Y=\hat{J}=\frac{X}{k}$ as an estimator for $J$. We show that for the above choice of $k$, $Y$ is accurate up to $\epsilon$ with probability of at least $\frac{7}{8}$.

\begin{theorem}[Simple Estimator] %Jackard 
$Pr(|Y-J| \leq \epsilon) \geq \frac{7}{8}$. Proof given in appendix. 
\label{l_thm_jackard_estimator_single_bit}
\end{theorem}
%\begin{proof}
%\end{proof}

Due to Theorem~\ref{l_thm_jackard_estimator_single_bit}, we can approximate $J$ with accuracy $\epsilon$ and confidence $\frac{7}{8}$ using a ``block fingerprint'' for $C_i$, composed of $m^{h_1, \phi_1}_i, \ldots, m^{h_k, \phi_k}_i$, where $h_1,\ldots,h_k$ are randomly constructed members of a $\gamma$-MWIF and $\phi_1,\ldots,\phi_k$ are chosen from the PWIF of hashes $\phi : [u] \rightarrow \{0,1\}$. %Theorem~\ref{l_thm_prob_collision_jackard_approx_mwif_single_bit}
We shows that it suffices to take $k=O(\frac{1}{\epsilon^2})$ to achieve this. Constructing each $h_i$ can be done by choosing $f,g$ using the base random construction and composing $h_i(x) = f(x) + i \cdot g(x)$. The base random construction chooses $f,g$ uniformly at random from $F_d$, the family of $d$-degree polynoms in $\Zp$, where $d = O(\log \frac{1}{\epsilon})$. % TODO Cite Indyk, here and up. 
This achieves a $\gamma$-MWIF where $\gamma = (1-\alpha) \cdot \epsilon = \frac{1}{2^{10}} \epsilon$. 
% TODO: Ely - optimize 2^20 to change d, then update below. 

%\subsection{Achieving a Desired Confidence}
%\label{l_sect_jackard_req_conf}
\paragraph{{Achieving a Desired Confidence}}
We combine several \emph{independent} fingerprints to increase the confidence to a desired level $\delta$. Section~\ref{l_sect_pairwise_block_const_conf} used a fingerprint of length $k$ to achieve a confidence of $\frac{7}{8}$. Consider taking $m$ fingerprints for each stream, each of length $k$. Given two streams, $i,j$, we have $m$ pairs of fingerprints, each approximating $J$ with accuracy $\epsilon$, and confidence $\frac{7}{8}$. Denote the estimators we obtain as $\hat{J}_1,\hat{J}_2,\ldots,\hat{J}_m$, and denote the \emph{median} of these values as $\hat{J}$. Consider using $m > \frac{32}{9} \ln \frac{1}{\delta}$ ``blocks''. 

\begin{theorem}[Median Estimator] %for Jackard
$Pr(|\hat{J}-J| \leq \epsilon) \geq 1-\delta$. Proof given in appendix. 
\label{l_thm_jackard_median_estimator}
\end{theorem}
%\begin{proof}
%Given in appendix. 
%\end{proof}

Due to Theorem~\ref{l_thm_jackard_median_estimator} to make sure that $|\hat{J} - J| \leq \epsilon$ it suffices to take $m > \frac{32}{9} \ln \frac{1}{\delta}$ fingerprints, each with $k=\frac{8.02}{\epsilon^2}$ hashes. In total, it is enough to take $\frac{32}{9} \ln \frac{1}{\delta} \cdot \frac{8.02}{\epsilon^2} \leq \frac{28.45 \ln \frac{1}{\delta}}{\epsilon^2}$ hashes. Thus, we use $O(\frac{\ln \frac{1}{\delta}}{\epsilon^2})$ hashes, storing a single bit per hash. 

\section{Fast Method for Computing the Fingerprint}
\label{l_sect_fast_compute}
We discuss speeding up the fingerprint computation. Consider computing the fingerprint for a set of $b$ items $X = \{x_1,\ldots,x_b\}$ where $x_i \in [u]$. The fingerprint is composed of $m$ ``block fingerprints'', where block $r$ is constructed using $k$ hashes $h^r_1,\ldots,h^r_k$, built using $2 \cdot d$ random coefficients in $\Zp$. 
%Each fingerprint ``chunk'' is constructed using $k$ hashes, $h_1,\ldots,h_k$ has defined above. 
The $i$'th location in the block is the minimal item in $X$ under $h_i$: $m_i = \arg \min_{x \in X} h_i(x)$, which is then hashed through a hash $\phi$ mapping elements in $[u]$ to a single bit. We show how to quickly compute the block fingerprint $(m_1,\ldots,m_k)$. A naive way to do this is applying $k \cdot b$ hashes to compute $h_i(x_j)$ for $i \in [k], j \in [b]$. The values $h_i(x_i)$ where $i \in [k], j \in [b]$ form a matrix, where row $i$ has the values $(h_i(x_1), \ldots, h_i(x_b))$, illustrated in Figure~\ref{fig:fingerprint-chunk}. 

\begin{figure}[ht]
\centering
\includegraphics[width=110mm,height=50mm]{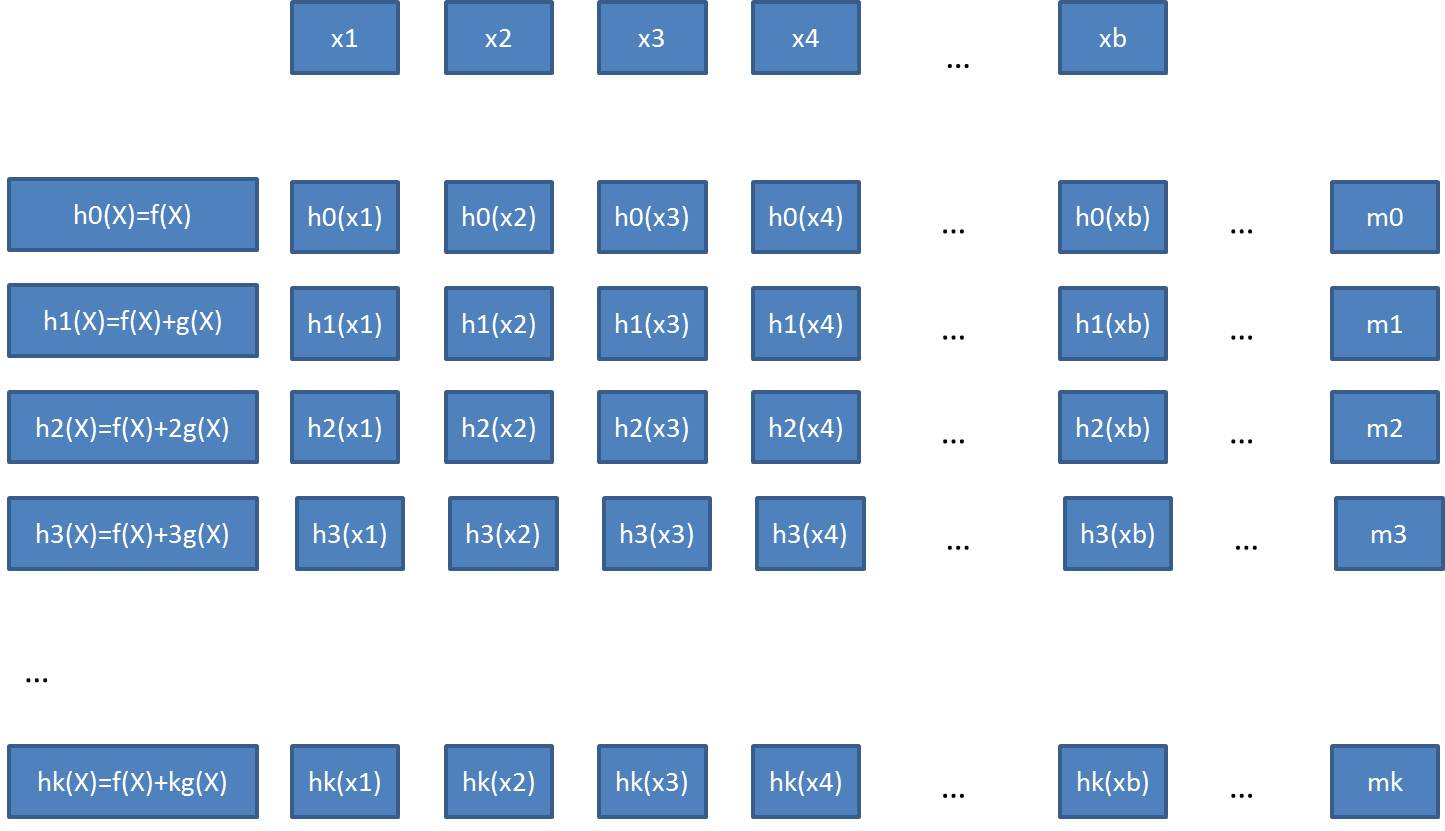}
\caption{A fingerprint ``chunk'' for a stream.}
\label{fig:fingerprint-chunk}
\end{figure}

Once all $h_i(x_j)$ values are computed for $i \in [k], j \in [b]$ , for each row $i$ we check for which column $j$ the row's minimal value occurs, and store $m_i = x_j$, as illustrated in the left of Figure~\ref{fig:fingerprint-min-row}. Thus, computing the fingerprint requires finding the minimal value across the rows (or more precisely, the value $x_j$ for the column $j$ where this minimal value occurs).
%, as shown in the next figures. 
To speed up the process, we use a method similar to the one discussed in~\cite{pavan2008range} as a building block. % TODO: Replace. 
Recall the hashes $h_i$ were defined as $h_i(x) = f(x) + i g(x)$ where $f(x),g(x)$ are $d$-degree polynomials with random coefficients in $\Zp$. 
%The methods discussed in~\cite{pavan2008range} 
Our algorithm is based on a procedure that gets a value $x \in[u]$ and a threshold $t$, and returns all elements in $(h_0(x),h_1(x),\ldots,h_{k-1}(x))$ which are smaller than $t$, as well as their locations. Formally, the method returns the index list $I_t = \{ i | h_i(x) \leq t \}$ and the value list $V_t = \{ h_i(x) | i \in I_t \}$ (note these are lists, so the $j$'th location in $V_t$, $V_t[j]$, contains $h_{I_t[j]}(x)$). We call this the \emph{column procedure}, and denote by $pr-small-loc(f(x),g(x),k,x,t)$ the function that returns $I_t$, and by $pr-small-val(f(x),g(x),k,x,t)$ the function that returns $V_t$ . We describe a certain implementation of these operations in Section~\ref{fast_pr_compute}. The running time of this implementation is $O(\log k + |I_t|)$, rather than the naive algorithm which evaluates $O(k)$ hashes. Thus, this procedure quickly finds small elements across columns (where by ``small'' we mean smaller than $t$). This is illustrated on the right of Figure~\ref{fig:fingerprint-min-row}. %TODO: Convert min to small. 

\begin{figure}[ht]
\centering
\includegraphics[width=110mm,height=40mm]{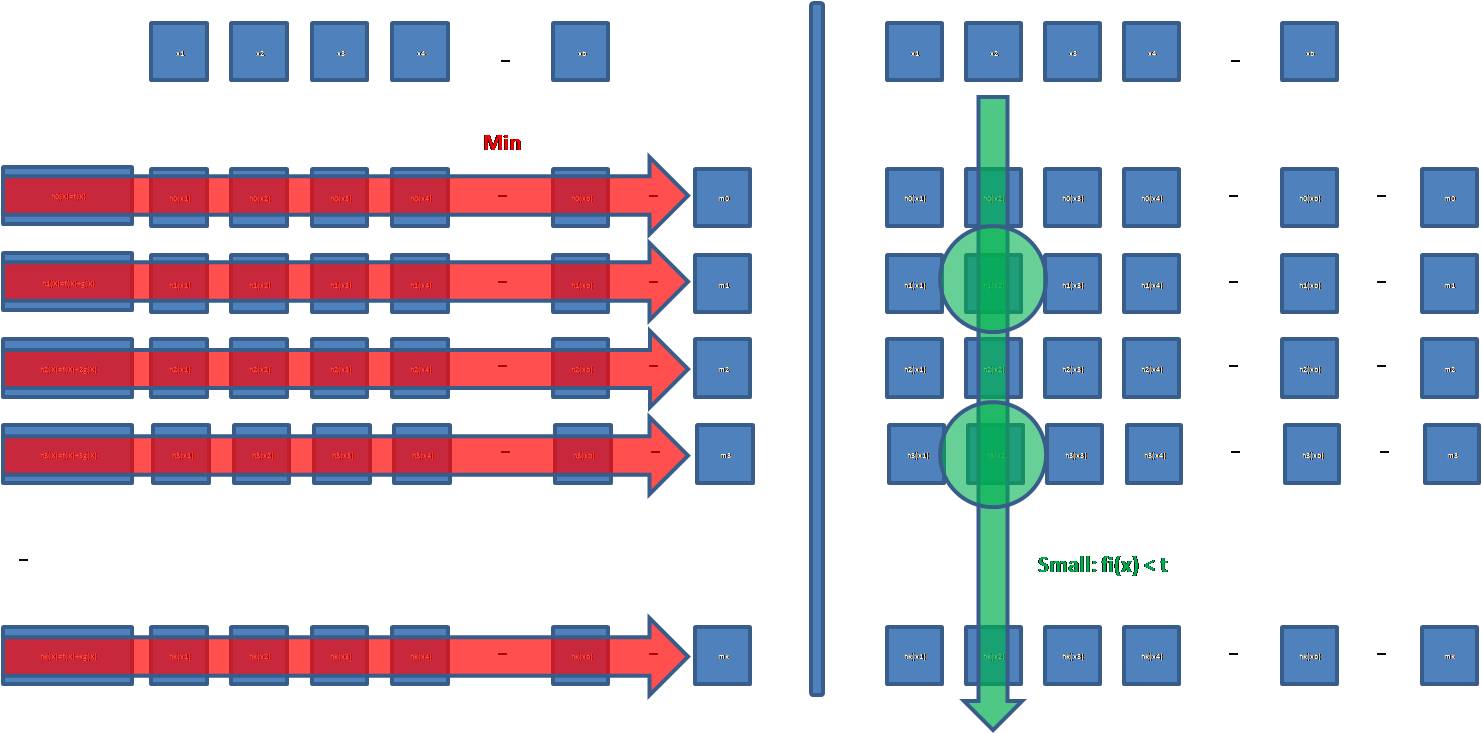}
\caption{Finding small elements across columns rather than minimal elements across rows}
\label{fig:fingerprint-min-row}
\end{figure}

%\begin{figure}[ht]
%\centering
%\includegraphics[width=110mm,height=40mm]{block-col-small.jpg}
%\caption{Basic fast procedure --- finding small values across a column.}
%\label{fig:fingerprint-min-col}
%\end{figure}

Roughly speaking, our algorithm maintains a bound for the minimal value for each row, and operates by going through the columns, finding the small values in each of them, and updating the bounds for the rows where these occur. 

$block-update( (x_1,\ldots,x_b), f(x), g(x), k, t):$
\begin{enumerate}
	\item Let $m_i = \infty$ for $i \in [k]$
	\item Let $p_i = 0$      for $i \in [k]$
	\item For $j = 1$ to $b$:
	\begin{enumerate}
		\item Let $I_t = pr-small-val (f(x),g(x),k,x_j,t)$
		\item Let $V_t = pr-small-loc (f(x),g(x),k,x_j,t)$		
		\item For $y\in I_t$: // Indices of the small elements
		\begin{enumerate}
			\item If $m_{I_t[y]} > V_t[y]$ // Update to row $x$ required
			\begin{enumerate}
				\item $m_{I_t[y]} = V_t[y]$
				\item $p_{I_t[y]} = x_j$
			\end{enumerate}
	  \end{enumerate}
	\end{enumerate}
\end{enumerate}

If our method updates $m_i, p_i$ for row $i$, once the procedure is done, $m_i$ indeed contains the minimal value in that row, and $p_i$ the column where this minimal value occurs, since if even a single update occurred then the row indeed contains an item that is smaller than $t$, so the minimal item in that row is smaller than $t$ and an update would occur for that item. On the other hand, if all the items in a row are bigger than $t$, an update would not occur for that row. The running time of the column procedure is $O(\log k + |I_t|)$, which is a random variable, that depends on the number of elements returned for that column, $|I_t|$. Denote by $L_j$ the number of elements returned for column $j$ (i.e. $|I_t|$ for column $j$). Since we have $b$ columns, the running time of the block update is $O(b \log k) + O(\sum_{j=1}^b L_j)$. The total number of returned elements is $\sum_{j=1}^b L_j$, which is the total number of elements that are smaller than $t$. We denote by $Y_t = \sum_{j=1}^b L_j$ the random variable which is the number of all elements in the block that are smaller than $t$. The running time of our block update is thus $O(b \log k + Y_t)$. 

The random variable $Y_t$ depends on $t$, since the smaller $t$ is the less elements are returned and the faster the column procedure runs. On the other hand, we only update rows whose minimal value is below $t$, so if $t$ is too low we have a high probability of having rows which are not updated correctly. We show that a certain compromise $t$ value allows achieving both a good running time of the block update, with a good probability of correctly computing the values for all the rows. 

\begin{theorem}%[Compute Time of Block Update]
Given the threshold $t=\frac{12 \cdot p \cdot l'}{b}$, where $l' = 80 + 2 \log \frac{1}{\epsilon}$ (so $l'=O(\log \frac{1}{\epsilon})$), the runtime of the $block-update$ procedure is $O(b\log\frac{1}{\epsilon} + \frac{1}{\epsilon^2}\log\frac{1}{\epsilon})$.
\label{l_thm_runtime_block_good_threshold}
\end{theorem}
\begin{proof}
Recall that to get a $\gamma$-MWIF (for $\gamma = \frac{1}{2^{10}} \epsilon$) we used $d = O(\log\frac{1}{\gamma})$ as the degree of the random polynoms $f,g$ in the base random construction, used to compose the $h_1,\ldots,h_k$ hashes. Examining the constant in the work of Indyk~\cite{indyk2001small} shows that the requirement is $d>80+2\log\frac{1}{\epsilon}$.
Denote $l' = 80 + 2 \log \frac{1}{\epsilon}$. Due to our choice of $d$ we have $d > l'$, so the hashes $h_1,\ldots,h_k$ were effectively chosen at random from an $l'$-wise independent family. Let $H$ be an $l'-wise$ independent family of hashes. Consider the following equation from~\cite{indyk2001small}, regarding $E_t$, the expected number of elements $x \in X$ such that $h(x) \leq t$ (i.e. elements that are smaller than $t$ under $h$ chosen at random from $H$):
%\begin{lemma}
$Pr[min_{x \in X} h(x) >  t] \leq 48 \left( \frac{6 \cdot l'}{E_t} \right) ^{(l'-1)/2}$. 
%\end{lemma}

\noindent When computing the fingerprint for the elements in $X$, we know $|X|$\footnote{We use this assumption for simplicity. If we don't know $|X|$, we can update the threshold $t$ online. We store all elements until we have $\frac{\log\frac{1}{\delta}}{\epsilon^2}$ elements. Then we set $t$ according to $b=2\frac{\log\frac{1}{\delta}}{\epsilon^2}$. We double $b$ by $2$ each time $|X|>b$ and update $t$ according to the new $b$.} and denoted $|X|=b$. Each $h_i$ is $\gamma$-MWIF, so $E_t = \frac{tb}{p}$. Now consider choosing $t=\frac{12 \cdot p \cdot l'}{b}$. Under this choice\footnote{Notice that this constant is only to bound the worst case usually in a block the maximum between the minimal values is about $l'$ moreover we can improve the running time if we drop from the sketch all the hash functions which there minimal value is to big.} of $t=\frac{12l'\cdot p}{b}$  we have $E_t = \frac{tb}{p} = 12l'$ and using the fact that $l' = 80 + 2 \log \frac{1}{\epsilon}$ the above lemma can be rewritten as:
$Pr[min_{x \in X} h(x) > t]<48 \left( \frac{6l'}{E_t} \right) ^{(l'-1)/2} = 48 \cdot \left( \frac{1}{2} \right) ^ {\frac{79}{2}} \cdot \left( \frac{1}{2} \right) ^ {2\log \frac{1}{\epsilon}} < \frac{1}{2^{33}} \cdot \epsilon^2 $.
There are $k$ rows, and by applying the union bound we obtain:
$Pr[\exists i \in [k] ( min_{x \in X} h_i(x) > t] < \frac{k \cdot \epsilon }{2^{33}} = \frac{8.02 \cdot \epsilon^2 }{2^{8.9} \cdot \epsilon^2} < \frac{1}{2^{29}} $.

We prove our algorithm runs in time $O(b\log\frac{1}{\epsilon} + \frac{1}{\epsilon^2}\log\frac{1}{\epsilon})$ with high probability. We have $kb$ random values, $h_1(x_1),\ldots, h_{k-1}(x_b)$, which are (at least) pairwise independent. Denote $Y_{i,j}$ the indicator variable of the event that $h_j(x_i) < t = \frac{12pl'}{b}$, and so $Pr[Y_{i,j}=1]=\frac{12l'}{b}$ and $E[Y_{i,j}=1]=\frac{12l'}{b}$. Then $Y=\sum_{i=0}^b\sum_{j=0}^{k-1} Y_{i,j}$. The running time of the algorithm is $O(b\log\frac{1}{\epsilon} + Y)$. We show that $Y=O(\frac{1}{\epsilon^2}\log\frac{1}{\epsilon})$ with high probability\footnote{We base our calculation on the pairwise independence of $Y_{i,j}$. Notice that $Y_{i,j}$ is more independent when running over $i$. Therefor in practice the constants are smaller.}. We obtain: $E[Y]=E[\sum_{i=0}^b\sum_{j=0}^{k-1} Y_{i,j}]=\sum_{i=0}^b\sum_{j=0}^{k-1} E[Y_{i,j}]=12 \cdot l' \cdot k$. We use the following lemma, proven in the appendix: $Var(Y) \leq E(Y)$, and using Chebychev's inequality obtain: $Pr[Y>11E(Y)]\le Pr[|Y-E(Y)|>10Var(Y)]<\frac{1}{100}$.
To guarantee the required run time in a worst case analysis, we can drop all the blocks which require too long to compute. This reduces our probability of success in each block from $\frac{7}{8}$ to at least $\frac{7}{8}-2^{-29}-\frac{1}{100}$ (The $2^{-29}$ factor is due to the probability that there exists a hash that gets a minimum value higher than $t$). Taking $4\log\frac{1}{\delta}$ blocks still obtains this probability. Overall the algorithm runs in time $O(b\log\frac{1}{\epsilon} + \frac{1}{\epsilon^2}\log\frac{1}{\epsilon})$ per block, or $O(\log\frac{1}{\delta}(b\log\frac{1}{\epsilon} + \frac{1}{\epsilon^2}\log\frac{1}{\epsilon}))$ for all blocks.%, as required. 
\end{proof}

\subsection{Computing The Minimal Elements of the Pseudo-Random Series}
\label{fast_pr_compute}

We give a recursive implementation of $pr-small-loc(f(x),g(x),k,x,t)$ and $pr-small-val(f(x),g(x),k,x,t)$, the procedures for computing $V_t$ and $I_t$. Recall the hashes $h_i$ were defined as $h_i(x) = f(x) + i g(x)$ where $f(x),g(x)$ are $d$-degree polynomials with random coefficients in $\Zp$. Consider a given element $x \in \Zp$ for which we attempt to find all the values (and indices) in $(h_0(x),h_2(x),\ldots,h_{k-1}(x))$ smaller than $t$.
%and the indices where they occur. 
Given $x$, we can evaluate $f(x), g(x)$ in time $O(d) = O(\log\frac{1}{\gamma})$\footnote{Using multipoint evaluation we can calculate it in amortized time $O(\log^2\log\frac{1}{\gamma})$. Moreover we can use other constructions for $d$-wise independent which can be evaluate in $O(1)$ time in the cost of using more space.}, and denote $a = f(x) \in \Zp$ and $b=g(x) \in \Zp$. Thus, we are seek all values in $\{ a \mod p,(a+b) \mod p, (a+2b) \mod p,\ldots, (a+(k-1) b)  \mod p \}$ smaller than $t$, and the indices $i$ where they occur. Consider the series $S = (s_1,\ldots,s_k)$ where $s_i = (a+ib) \mod p$ and $i = \{0,1,\ldots,k-1 \}$. We denote the arithmetic series $a+bi \mod p$ for $i \in \{0,1,\ldots,k-1\}$ as $S(a,b,k,p)$, so under this notation $S = S(a,b,k,p)$. 

Given a value we can find the index where it occurs, and vice versa. To compute the value for index $i$, we compute $(a+ib) \mod p$. To compute the index $i$ where a value $v$ occurs, we solve $v=a+ib$ in $\Zp$ (i.e. $i=\frac{v-a}{b}\mod p$).
%, or equivalently we have the equation $i=\frac{v-a}{b}$ in $\Zp$ 
This can be done in $O(\log p)$ time using Euclid's algorithm. Note we compute $b^{-1}$ in $\Zp$ only once to transform \emph{all} values to generating indices\footnote{We can store a table of inverse to further reduce processing time. If the required memory for the table is unavailable, we can do the computation in $F_{p^c}$ for smaller $p$ and store table of size $p$ and then calculating the inverse requires $O(c\log c)$ time. Notice that we can easily take $c<\log_{\frac{\log\frac{1}{\delta}}{\epsilon^2}} u$ which will probably be less then $\log\frac{1}{\epsilon}$}. We call a location $i$ where $s_i < s_{i-1}$ a \emph{flip location}. The first index 
%has no predecessor, and we call it 
is a flip location if $a-b \mod p > a$. First, consider the case $b < \frac{p}{2}$. If $s_i$ is a flip location, we have $s_{i-1} < p$ but $s_{i-1} + b > p$, so $s_i < b$. Also, since $b < \frac{p}{2}$ there is at least one location which is \emph{not} a flip location between any two flip locations. Given $S=S(a,b,k,p)$, denote by $f(S)$ the flip locations in $S$.

\begin{lemma}[Flip Locations Are Small]
\label{l_lem_flip_small}
When $b < \frac{p}{2}$, at most $\frac{k}{2}$ elements are flip locations, and all elements that are smaller than $b$ are flip locations.
\end{lemma}
\begin{proof}
Note that the non-flip locations between any two flip locations are monotonically increasing. Any flip location has a value of at most $b$, since the element before a flip location is smaller than $p$ (modulo $p$), and adding $b$ to it exceeds $p$, but through this addition it is impossible to exceed $p$ by more than $b$. 
\end{proof}

We denoted by $f(S)$ the flip locations of $S$. Denote $f_0(S)=f(S)$. Denote by $f_1(S)$ all elements that occur directly after a flip location, $f_2(S)$ all elements that occur exactly two places after the closest flip locations (i.e they cannot be flip locations) and by $f_i(S)$ all elements that occur $i$ places after the closest flip location. 

\begin{lemma}[Element Comparison]
\label{l_lem_element_comp}
When $b < \frac{p}{2}$, if $x \in f_i(S)$ and $y \in f_j(S)$ where $i > j$, then $x > y$.
\end{lemma}
\begin{proof}
All flip locations have a value of at most $b$. Due to Lemma~\ref{l_lem_flip_small}, a location directly after a flip location is not a flip location, and is thus bigger than the flip location before it by exactly $b$, and is thus greater than $b$. Thus any element in $f_1(S)$ must be greater than any element in $f_0(S)$. Using the same argument, we see that any element in $f_2(S)$ is greater than any element in $f_1(S)$ and so on. A simple induction completes the proof. 
\end{proof}

The first flip location is $\ceil{\frac{p-a}{b}}$, as to exceed $p$ we add $b$ $\ceil{\frac{p-a}{b}}$ times. Also, the number of flip locations is $\floor{\frac{a+bk}{p}}$. %The functions that compute these are given in the following pseudo-code:
%$$first-flip(a,b,p,k): \ceil{\frac{p-a}{b}}$$
%$$num-flips(a,b,p,k): \floor{\frac{a+bk}{p}}$$
Denote the first flip location as $j=\ceil{\frac{p-a}{b}}$, with value $a' = (a+jb) \mod p$. Denote $b' = (b-p) \mod b$ and the number of flip locations as $k' = \floor{\frac{(a+bk)}{p}}$. The flip locations are known to also be an arithmetic progression~\cite{pavan2008range} \footnote{See Lemma 2 page 11.}. 

\begin{lemma}[Flip Locations Arithmetic Progression]
\label{l_lem_flip_arit_prog}
The flip locations of $S=S(a,b,k,p)$ are also an arithmetic progression $S'=(a',b',k', b)$. 
\end{lemma}

Given the above lemmas, we can search for the elements smaller than $t$, by examining the flip locations series in recursion. If case $b<t$, given $q=\ceil{t}{b}$, due to Lemma~\ref{l_lem_element_comp} $f(S),f_1(S),\ldots f_{q-1}(S)$ are smaller then $t$, and all of their elements must be returned. We must also scan $f_q(S)$ and also return all the elements of $f_q(S)$ which are smaller then $t$. This additional scan requires $O(|f_q(S)|)$ time $|f_q(S)|\leq |f(S)|$. Thus this case of $b < t$ examines $O(|I_t|)$ elements. Due to Lemma~\ref{l_lem_flip_small}, if $b > t$, all non-flip locations are bigger than $b$ and thus bigger than $t$, and thus we must only consider the flip-locations as candidates. Using Lemma~\ref{l_lem_flip_arit_prog} we can scan the flip locations recursively by examining the arithmetic series of the flip locations. If at most half of the elements in each recursion are flip locations, this results in a logarithmic running time. However, if $b$ is high more than half the elements are flip locations. For the case where $b > \frac{p}{2}$ we can examine the same flip-location series $S'$, in reverse order. The first element in the reversed series would be the last element of the current series, and rather than progressing in steps of $b$, we progress in steps of $p-b$. This way we obtain exactly the same elements, but in reverse order. However, in this reversed series, at most half the elements are flip locations. 
%Now note that in each recurssion step, examining the flip location series results in a lower upper bound on the size of the elements. The original series is in the range $[p]$, the next series (the flip location series) is in the range $[b]$, the next one in the range $[b']=b_1$ and so on. Denote the upper bound obtained in the $i$'th series as $u_i$. Once the recurssion reaches a series $S'$ where $u_i < t$, we know all flip locations $f_0(S')$ are smaller than $t$, then proceed to examine $f_1(S')$, then $f_2(S')$ until $f_q(S')$ where $qb > t$. Denote by $e_t$ the number of elements smaller than $t$. Thus, in the last step of the recurssion we examine at most $2 e_t$ elements. We are now ready to present our algorithm.
%%% 
The following procedure implements the above method. It finds elements smaller then $t$ in time $O(\log k) = O(\log\frac{1}{\epsilon}+|I_t|)$ where $|I_t|$ is the number of such values. Given the returned indices, we get the values in them. We use the same $b$ for all $|I_t|$, so this can be done in time $O(c\log c + |I_t|)$ (Usually $c$ is a constant). %TODO: What's C?

\noindent $ps-min(a,b,p,k,t):$
\begin{enumerate}
	\item if $b<t$: 
	\begin{enumerate}
  	\item $V_t=[]$
    \item if $a < t$ then $V_t=V_t+[a+ib \text{ for i in range } (\ceil{\frac{t-a}{b}})]$
    \item $j=\ceil{\frac{p-a}{b}}$ // First flip (excluding first location)
    \item while $j<k$:
    \begin{enumerate}
    	\item $v=(a+jb) \mod p$
      \item while $j<k$ and $v<t$:
      \begin{enumerate}
      	\item $V_t$.append(v)
        \item $j=j+1$
        \item $v=v+b$
    \end{enumerate}
    \item $j=j+\ceil{\frac{p-v}{b}}$ //next flip location
    \item return list1
  \end{enumerate}
  \item if $b>\frac{p}{2}$ then return $f((a + (k-1) \cdot b) \mod p, p-b, p, k, t)$
  \item $j=\ceil{\frac{p-a}{b}}$
  \item $new_k=\floor{\frac{a+bk}{p}}$
    \item if $a<b$ then $j=0$ and $new_k=new_k+1$// calculate the first flip location and the number of flip locations
  \item return $f( (a+jb) \mod p, -p \mod b, b, new_k, t )$
  \end{enumerate}
\end{enumerate}

\section{Conclusions}
\label{sec:conclusions}

We have presented a fast method for computing fingerprints of massive datasets, based on pseudo-random hashes. We note that although we have examined the Jackard similarity in detail, the exact same technique can be used for any fingerprint which is based on minimal elements under several hashes. % TODO - cite several of these. 
Thus we have described a general technique for exponentially speeding up computation of such fingerprints. Our analysis has used fingerprints using a single bit per hash. We have shown that even for these small fingerprints which can be quickly computed, the required number of hashes is asymptotically similar to previously known methods, and is logarithmic in the required confidence and polynomial in the required accuracy. Several directions remain open for future research. Can we speed up the fingerprint computation even further? Can similar techniques be used for computing fingerprints that are not based on minimal elements under hashes? % Ely- are there really examples of fingerprints where this technique would not work directly?
%Finally, can we prove lower bounds on the required number of hashes to achieve a required accuracy and confidence using this family of MWIFs? XXXXXX THERE IS A LOWER BOUND

\bibliographystyle{plain} 
\bibliography{fprf}

\begin{thebibliography}{10}

\bibitem{alon1999space}
N.~Alon, Y.~Matias, and M.~Szegedy.
\newblock {The Space Complexity of Approximating the Frequency Moments}.
\newblock {\em Journal of Computer and System Sciences}, 58(1):137--147, 1999.

\bibitem{bachrach-fingerprinting}
Y.~Bachrach and R.~Herbrich.
\newblock {Fingerprinting Ratings For Collaborative Filtering—Theoretical and
  Empirical Analysis}.
\newblock 2010.

\bibitem{bachrach2009spire}
Y.~Bachrach, R.~Herbrich, and E.~Porat.
\newblock {Sketching Algorithms for Approximating Rank Correlations in
  Collaborative Filtering Systems}.
\newblock In {\em String Processing and Information Retrieval}, pages 344--352.
  Springer, 2009.

\bibitem{bachrach2009sketching}
Y.~Bachrach, E.~Porat, and J.S. Rosenschein.
\newblock {Sketching techniques for collaborative filtering}.
\newblock {\em IJCAI 2009}.

\bibitem{broder1998resemblance}
A.Z. Broder.
\newblock {On the resemblance and containment of documents}.
\newblock {\em In Compression and Complexity of Sequences (SEQUENCES'97}, 1998.

\bibitem{broder2000min}
A.Z. Broder, M.~Charikar, A.M. Frieze, and M.~Mitzenmacher.
\newblock {Min-wise independent permutations}.
\newblock {\em Journal of Computer and System Sciences}, 60(3):630--659, 2000.

\bibitem{broder2003derandomization}
A.Z. Broder, M.~Charikar, and M.~Mitzenmacher.
\newblock {A derandomization using min-wise independent permutations}.
\newblock {\em Journal of Discrete Algorithms}, 1(1):11--20, 2003.

\bibitem{cohen2007sketching}
E.~Cohen, N.~Duffield, H.~Kaplan, C.~Lund, and M.~Thorup.
\newblock {Sketching unaggregated data streams for subpopulation-size queries}.
\newblock In {\em Proceedings of the twenty-sixth ACM SIGMOD-SIGACT-SIGART
  symposium on Principles of database systems}, page 262. ACM, 2007.

\bibitem{cohen2007summarizing}
E.~Cohen and H.~Kaplan.
\newblock {Summarizing data using bottom-k sketches}.
\newblock In {\em Proceedings of the twenty-sixth annual ACM symposium on
  Principles of distributed computing}, page 234. ACM, 2007.

\bibitem{datar2002estimating}
M.~Datar and S.~Muthukrishnan.
\newblock {Estimating rarity and similarity over data stream windows}.
\newblock {\em Algorithms—ESA 2002}, pages 323--335.

\bibitem{hoeffding:1963}
Wassily Hoeffding.
\newblock Probability inequalities for sums of bounded random variables.
\newblock {\em Journal of the American Statistical Association},
  58(301):13--30, 1963.

\bibitem{indyk2001small}
P.~Indyk.
\newblock {A Small Approximately Min-Wise Independent Family of Hash
  Functions}.
\newblock {\em Journal of Algorithms}, 38(1):84--90, 2001.

\bibitem{indyk2006stable}
P.~Indyk.
\newblock {Stable distributions, pseudorandom generators, embeddings, and data
  stream computation}.
\newblock {\em Journal of the ACM (JACM)}, 53(3):323, 2006.

\bibitem{krauthgamer2008greedy}
R.~Krauthgamer, A.~Mehta, V.~Raman, and A.~Rudra.
\newblock {Greedy list intersection}.
\newblock In {\em IEEE 24th International Conference on Data Engineering, 2008.
  ICDE 2008}, pages 1033--1042, 2008.

\bibitem{li2010b}
P.~Li and C.~Konig.
\newblock {b-Bit minwise hashing}.
\newblock In {\em Proceedings of the 19th international conference on World
  wide web}, pages 671--680. ACM, 2010.

\bibitem{mulmuley1996randomized}
K.~Mulmuley.
\newblock {Randomized geometric algorithms and pseudorandom generators}.
\newblock {\em Algorithmica}, 16(4):450--463, 1996.

\bibitem{patrascu:k}
M.~Patrascu and M.~Thorup.
\newblock {On the k-Independence Required by Linear Probing and Minwise
  Independence}.

\bibitem{pavan2008range}
A.~Pavan and S.~Tirthapura.
\newblock {Range-efficient counting of distinct elements in a massive data
  stream}.
\newblock {\em SIAM Journal on Computing}, 37(2):359--379, 2008.

\end{thebibliography}

\section{Appendix: Proofs}
\label{l_sect_proofs}

The proof of Theorem~\ref{l_thm_prob_collision_jackard}:
%\begin{theorem}[Jackard and MWIF Collision Probability]
$Pr_{h \in H} [m^h_i = m^h_j] = J_{i,j}$.
%\end{theorem}
\begin{proof}
Denote $x=J_{1,2}$. The set $C_i \cup C_j$ contains three types of items: items that appear \emph{only} in $C_i$, items that appear \emph{only} in $C_j$, and items that appear in $C_i \cap C_j$. When an item in $C_i \cap C_j$ is minimal under $h$, i.e., for some $a \in C_i \cap C_j$ we have $h(a) = min_{x \in C_1 \cup C_2} h(x)$, we get that $min_{x \in C_i} h(x) = min_{x \in C_j} h(x)$. On the other hand, if for some $a \in C_i \cup C_j$ such that $a \notin C_i \cap C_j$ we have $h(a) = min_{x \in C_1 \cup C_2} h(x)$, the probability that $min_{x \in C_i} h(x) = min_{x \in C_j} h(x)$ is negligible \footnote{Such an event requires that two \emph{different} items, $x_i \in C_i$ and $x_j \in C_j$ would be mapped to the same value $h^* = h(x_i) = h(x_j)$, and that this value would also be the minimal value obtained when applying $h$ to both all the items in $C_i$ and in $C_j$. As discussed in~\cite{indyk2001small}, the probability for this is negligible when the range of $h$ is large enough.}. Since $H$ is MWIF, any element in $C = C_i \cup C_j$ is equally likely to be minimal under $h$. However, only elements in $I = C_i \cap C_j$ would result in $m^h_i = m^h_j$. Thus $Pr_{h \in H} [m^h_i = m^h_j] = \frac{1}{|C_i \cup C_j|} \cdot |C_i \cap C_j| = \frac{|C_i \cap C_j|}{|C_i \cup C_j|} = J_{i,j}$.
\end{proof}

The proof of Theorem~\ref{l_thm_jackard_estimator_single_bit} (Simple Estimator for Jackard With Single Bit Per Hash):
%\begin{theorem}[Simple Estimator for Jackard With Single Bit Per Hash]
$Pr(|Y-J| \leq \epsilon) \geq \frac{7}{8}$.  
%\end{theorem}
\begin{proof}
Our proof uses Chebychev's inequality:
$$ Pr(|X-E(X)| \geq \epsilon) \leq \frac{Var(X)}{\epsilon^2} $$

We have: 
$$E(X)=E(\sum_{l=1}^k X_l)=\sum_{l=1}^k E(X_l) = k \cdot (J \pm \gamma)$$ 
$$        (J-\gamma) \leq E(Y) \leq         (J+\gamma)$$
\noindent We now bound $Var(X)$:
\begin{equation}
\begin{split}
Var(X) &= E(X^2) - E^2(X) \\
 &=E((\sum_{l=1}^k X_l)^2)-E^2(\sum_{l=1}^k X_l)\\
 &=E(\sum_{l=1}^k X_l^2+2\sum_{i\neq j} X_i X_j) -(E(\sum_{l=1}^k X_l))^2\\
 &=\sum_{l=1}^k E(X_l^2)+ 2\sum_{i\neq j} E(X_i X_j)-(\sum_{l=1}^k E(X_l))^2\\
 &= \sum_{l=1}^k E(X_l^2)+ 2\sum_{i\neq j} E(X_iX_j)-(\sum_{l=1}^k E(X_l)^2 +2\sum_{i\neq j} E(X_i)E(X_j))\\
 &=\sum_{l=1}^k E(X_l^2)+ 2\sum_{i\neq j} E(X_i)E(X_j)-(\sum_{l=1}^k E(X_l)^2 +2\sum_{i\neq j} E(X_i)(X_j))\\
 &=\sum_{l=1}^k E(X_l^2)-\sum_{l=1}^k E(X_l)^2 \leq k
\end{split}
\end{equation}

\noindent We use this to bound $Var(Y)$:
$$Var(Y)=Var(\frac{1}{k} \cdot X)=\frac{1}{k^2} Var(X) \leq \frac{1}{k^2} \cdot k \leq = \frac{1}{k}$$

\noindent Using Chebychev's inequality we get that:
$$ Pr(|Y-E(Y) | > \beta ) \leq \frac{Var(Y)}{\beta^2} \leq \frac{1}{k \cdot \beta^2} $$

\noindent Denote $\alpha = \frac{2^{10}-1}{2^{10}}$. Let $\beta = \alpha \cdot \epsilon$, so we obtain:

\noindent Thus using our choice of $k=\frac{8.02}{\epsilon^2}$ and $\beta = \alpha \cdot \epsilon$ (and noting that $J \leq 1, \epsilon \leq 1$) we have:

$$ Pr(|Y-E(Y) | > \beta ) \leq \frac{1}{k {\beta}^2} = 
\frac{1}{k \cdot \alpha^2 \cdot \epsilon^2} =
\frac{1}{8.0001} \leq \frac{1}{8} $$
\end{proof}

Proof of Theorem~\ref{l_thm_jackard_median_estimator} (Median Estimator for Jackard):
%\begin{theorem}[Median Estimator for Jackard]
$Pr(|\hat{J}-J| \leq \epsilon) \geq 1-\delta$.  
%\end{theorem}
\begin{proof}
We use Hoeffding's inequality~\cite{hoeffding:1963}. Let $X_1,\ldots , X_n$ be independent random variables, where all $X_i$ are bounded so that $X_i \in [a_i, b_i]$, and let $X = \sum_{i=1}^n X_i$. Hoeffding's inequality states that:
$$\Pr(X - \mathrm{E}[X] \geq n \epsilon) \leq \exp \left( -\frac{2\,n^2\,\epsilon^2}{\sum_{i=1}^n (b_i - a_i)^2} \right)$$

We say that the estimator $\hat{J}_l$ is \emph{good} if $|\hat{J}_l-J| \leq \epsilon$ and that $\hat{J}_l$ is \emph{bad} if $|\hat{J}_l-J| > \epsilon$. Each estimator $\hat{J}_l$ is bad with probability of $p \leq \frac{1}{8}$. Consider the random variable $X_l$ where $X_l=1$ if $\hat{J}_l$ is bad, and $X_l=0$ if $\hat{J}_l$ is good. We have $Pr(X_l=1) = p \leq \frac{1}{8}$, so $E(X_l) = p \leq \frac{1}{8}$. Denote $X=\sum_{l=1}^m X_l$, so $E(X)=m \cdot p \leq \cdot m \cdot \frac{1}{8}$. We now note that the $\hat{J}$ can be bad only if at least half the estimators $\hat{J}_1,\ldots,\hat{J}_m$ are bad, or in other words, when $X \geq \frac{m}{2}$.

The $X_l$'s are independent, since for any $x,y$ the hashes used to obtain the $\hat{J}_x$ are independent of the hashes used to obtain the $\hat{J}_y$. Since $p \leq \frac{1}{8}$ we have:
$$\Pr(X \geq \frac{m}{2}) \leq \Pr(X \geq (\frac{3}{8} + p) \cdot m) = \Pr(X - mp \geq \frac{3}{8} m)$$
\noindent However, $E(X) = mp$, so using Hoeffding's inequality, we require that $\Pr(X \geq \frac{m}{2}) \leq \delta$:
$$\Pr(X \geq \frac{m}{2}) \leq \Pr(X-mp \geq \frac{3}{8} m) \leq \exp (-2m \cdot \frac{9}{64}) \leq \delta$$

Extracting $m$ we obtain that we require: 
$$m > \frac{32}{9} \ln \frac{1}{\delta}$$
\end{proof}

Proof of the lemma in Theorem~\ref{l_thm_runtime_block_good_threshold}:

\begin{lemma}
Let $Y=\sum_{i=0}^b\sum_{j=0}^{k-1} Y_{i,j}$ in Theorem~\ref{l_thm_runtime_block_good_threshold}. Then $Var[Y] \leq E[Y]$. 
\end{lemma}
\begin{proof}
\begin{equation}
\begin{split}
Var[Y] &= E[Y^2]-E^2[Y] = E[(\sum_{i=0}^b\sum_{j=0}^{k-1} Y_{i,j})^2]-E^2[\sum_{i=0}^b\sum_{j=0}^{k-1} Y_{i,j}] \\
 &= E[\sum_{i=0}^b\sum_{j=0}^{k-1} Y_{i,j}^2+2\sum_{i'\neq i}\sum_{j'\neq j} Y_{i,j} Y_{i',j'}] -(\sum_{i=0}^b\sum_{j=0}^{k-1} E[Y_{i,j}])^2 \\
 &= \sum_{i=0}^b\sum_{j=0}^{k-1} E[Y_{i,j}^2]+2\sum_{i'\neq i}\sum_{j'\neq j} E[Y_{i,j} Y_{i',j'}] -(\sum_{i=0}^b\sum_{j=0}^{k-1} E[Y_{i,j}]^2+2\sum_{i'\neq i}\sum_{j'\neq j} E[Y_{i,j}][Y_{i',j'}]) \\
 &= \sum_{i=0}^b\sum_{j=0}^{k-1} E[Y_{i,j}^2]+2\sum_{i'\neq i}\sum_{j'\neq j} E[Y_{i,j}][Y_{i',j'}] -(\sum_{i=0}^b\sum_{j=0}^{k-1} E[Y_{i,j}]^2+2\sum_{i'\neq i}\sum_{j'\neq j} E[Y_{i,j}][Y_{i',j'}]) \\
 &= \sum_{i=0}^b\sum_{j=0}^{k-1} E[Y_{i,j}^2]-\sum_{i=0}^b\sum_{j=0}^{k-1} E[Y_{i,j}]^2 = \sum_{i=0}^b\sum_{j=0}^{k-1} E[Y_{i,j}]-\sum_{i=0}^b\sum_{j=0}^{k-1} E[Y_{i,j}]^2 \\
 &= E[Y]-\sum_{i=0}^b\sum_{j=0}^{k-1} E[Y_{i,j}]^2\le E[Y]
\end{split}
\end{equation}

\end{proof}

\end{document}